\documentclass[final,3p]{elsarticle} 	
\journal{Annual Reviews in Control}
\usepackage{amsmath,amssymb,bm,bbm,mathrsfs,amscd}
\usepackage{amsthm,subfigure}
\usepackage{dsfont, color}
\usepackage{graphicx}
\usepackage{epsfig}
\def\qed{\hfill \vrule height 7pt width 7pt depth 0pt}
\def\beq{\begin{equation}}
\def\eeq{\end{equation}}

\newtheorem{theorem}{Theorem}
\newtheorem{definition}[theorem]{Definition}
\newtheorem{proposition}[theorem]{Proposition}

\newtheorem{example}{Example}
\theoremstyle{remark}
\newtheorem{remark}[theorem]{Remark}
\newcommand{\ds}{\displaystyle}

\newcommand{\ba}{\begin{array}}
\newcommand{\ea}{\end{array}}

\newcommand{\be}{\begin{equation}}
\newcommand{\ee}{\end{equation}}

\newcommand{\mc}{\mathcal}

\newcommand{\ov}{\overline}

\newcommand{\1}{\mathbbm{1}}

\newcommand{\R}{\mathbb{R}}

\newcommand{\F}{\mathbb{F}}

\newcommand{\summ}{\sum\limits}

\newcommand{\de}{\mathrm{d}}

\newcommand{\se}{\text{ if }}

\DeclareMathOperator{\sgn}{sgn}

\def\1{\mathds{1}}

\def\R{\mathbb{R}}

\def\diag{{\rm diag}\,}

\def\F{{\cal F}}

\def\x{{\bf x}}

\begin{document}
\begin{frontmatter}
\title{On resilient control of dynamical flow networks
}



\author{Giacomo Como\tnotetext[authornote]{The author is a member of the excellence centres LCCC and ELLIIT. His research has been supported by the Swedish Research Council through a Project Research Grant. 
}
}
\address{G.~Como is with the Department of Automatic Control, Lund University, BOX118, SE-22100 Lund, Sweden  and with the Lagrange Department of Mathematical Sciences, Politecnico di Torino, Corso Duca degli Abruzzi 24, 10129, Torino, Italy. Email:  {\tt\small giacomo.como@control.lth.se}, {\tt\small giacomo.como@polito.it}.} %









\begin{abstract}
Resilience has become a key aspect in the design of contemporary infrastructure networks. This comes as a result of ever-increasing loads, limited physical capacity, and fast-growing levels of interconnectedness and complexity due to the recent technological advancements. The problem has motivated a considerable amount of research within the last few years, particularly  focused on the dynamical aspects of network flows, complementing more classical static network flow optimization approaches. 

In this tutorial paper, a class of single-commodity first-order models of dynamical flow networks is considered. A few results recently appeared in the literature and dealing with stability and robustness of dynamical flow networks are gathered and originally presented in a unified framework. In particular, (differential) stability properties of monotone dynamical flow networks are treated in some detail, and the notion of margin of resilience is introduced as a quantitative measure of their robustness. While emphasizing methodological aspects ---including structural properties, such as monotonicity, that enable tractability and scalability--- over the specific applications, connections to well-established road traffic flow models are  made.
\end{abstract}
\begin{keyword}  dynamical flow networks \sep transportation networks \sep compartmental systems \sep distributed control \sep network resilience \sep network flow control \sep dynamical routing \sep robust control \sep monotone systems.



\end{keyword}

\end{frontmatter}
\section{Introduction}


As critical infrastructure networks, such as transport and energy, are being utilized closer and closer to their capacity limits, the complex interaction between physical systems, cyber layers, and human decision makers has created new challenges in simultaneously achieving efficiency and reliability. The recent technological advancements in terms of smart sensors, high-speed communication, and real-time decision capabilities have exacerbated the large-scale interconnected nature of these systems, and increased both the potential gains associated to their optimization and their inherent systemic risks. In fact, while designed to perform well under normal operation conditions, such complex systems tend to exhibit critical fragilities in response to unforeseen disruptions. Even if simply started from small local perturbations, such disruptions have the potential to build up through cascading mechanisms driven by the interconnected dynamics of the infrastructure network, possibly leading to detrimental systemic effects. The term \emph{resilience} refers to the ability of these systems ``to plan and prepare for, absorb, respond to, and recover from disasters and adapt to new conditions'' (definition by the US National Academy of  Sciences  \cite{NAS-2012}).

Whilst static network flow optimization has long been regarded as a fundamental design paradigm for infrastructure systems and represents a central area of mathematical programming \cite{Ahuja.ea:1993,Bertsekas:1998,Whittle:2007}, there is an increasing awareness that the full potential of the emerging technologies and the nature of the associated systemic risks can only be understood by developing systems modeling, robustness analysis, and control synthesis within a dynamical framework. This recognition is stimulating considerable interest in the control systems community. In this tutorial paper, based on a semi-plenary lecture given by the author at the 22nd International Symposium on the Mathematical Theory of Networks and Systems, a few recent results on stability and robustness of dynamical flow networks are presented. While emphasizing methodological aspects ---in particular, structural properties enabling tractability and scalability of the considered models--- over the specific applications, this paper also makes connections to well-established road traffic flow models.

Our focus is on first-order models of \emph{dynamical flow networks}, describing the flow of mass among a finite set of interconnected cells. Such dynamical systems have sometimes been referred to as \emph{compartmental systems} in some of the literature \cite{Jacquez.Simon:93,Walterand.Contreras:99}. Our main interest is on \emph{nonlinear} dynamical flow networks, with nonlinearities especially accounting for \emph{congestion} effects. Special attention is devoted to \emph{demand} and \emph{supply} constraints limiting, respectively, the maximum outflow from and the maximum inflow in the cells as a function of their current mass, as in Daganzo's cell transmission model for road traffic network flows \cite{Daganzo:94,Daganzo:95}. We introduce a class of \emph{monotone} dynamical flow networks, characterized by structural properties of the dependence of the flow variables on the network state. We show that, while this class is large enough to encompass many examples of applicative interest, the system structure of monotone dynamical flow networks is such that their dynamic behaviors, and especially stability and robustness properties, are analyzable in a tractable and scalable way. In particular, this paper: (i) presents results relating the (differential) stability of (nonlinear) monotone dynamical flow networks to graph-theoretical properties; (ii) introduces the notion of \emph{margin of resilience} as a measure of their robustness against exogenous perturbations; and (iii) studies a class of \emph{locally responsive feedback} routing and flow control policies that are able to achieve the maximum possible margin of resilience for a given network topology in spite of relying on local information only and requiring no global knowledge of the network. 

The remainder of this paper is organized as follows. In Section \ref{sec:linear0}, we first introduce dynamical flow networks.  In Section \ref{sec:linear}, we focus on their simplest instance, affine dynamical flow networks, for which we gather some results relating their (global, exponential) stability to outflow-connectivity properties of the network topology. In Section \ref{sec:nonlinear-stability}, we define the notion of monotone dynamical flow networks, present stability results that can be deduced for this system structure, and show how they can be applied to dual ascent dynamics for static convex network flow optimization. In Section \ref{sec:demand-supply}, we introduce nonlinear dynamical flow networks with demand and supply constraints and show how several examples of network flow dynamics from the literature that can be fit in this framework also belong to the class of monotone dynamical flow networks (either globally or locally), so that the stability results of Section \ref{sec:nonlinear-stability} can be successfully applied.  In Section \ref{sec:resilience}, we study robustness of nonlinear dynamical flow networks with respect to perturbations of the demand functions (hence, of the flow capacities) as well as of the external inflows. We introduce the notion of margin of resilience as a quantitative measure of robustness and compute the margin of resilience of different classes of distributed routing and flow control policies. 

We end this introductory section by gathering some notational conventions to be adopted throughout the paper. The sets of real numbers and of nonnegative real numbers are denoted by $\R$ and $\R_+$, respectively. The all-one vector is denoted by $\1$, the all-zero vector simply by $0$, the identity matrix by $I=\diag(\1)$, and the transpose of a matrix $M$ by $M^T$. Inequalities between vectors are meant to hold true entrywise, i.e., if $a,b\in\R^n$, then $a\ge b$ means that $a_i\ge b_i$ for all $i=1,\ldots,n$. A square matrix $M$ is called: \emph{nonnegative} if all its entries are nonnegative; \emph{Metzler} if all its non-diagonal entries are nonnegative, i.e., $M_{ij}\ge0$ for all $i\ne j$ \cite{Berman.Plemmons:94}; (row) \emph{diagonally dominant} if $|M_{ii}|\ge\sum_{j\ne i}|M_{ij}|$; \emph{compartmental} if it is Metzler and $M\1\le0$; 
\emph{Hurwitz} if all its eigenvalues have negative real part; and \emph{substochastic} if it is nonnegative and such that $M\1\le\1$, i.e., its rows all sum up to less than or equal to $1$. A directed graph, shortly \emph{digraph}, is the pair $(\mc V,\mc E)$ of a finite node set $\mc V$ and a link set $\mc E\subseteq\mc V\times\mc V$, whereby links $(i,j)\in\mc E$ are interpreted as pointing from node $i$ to node $j$. A length-$l$ \emph{path} from a node $i$ to a node $j$ in a digraph $\mc G=(\mc V,\mc E)$  is a sequence of nodes $\{v_0,v_1,\ldots,v_l\}\subseteq\mc V$ such that $v_0=i$, $v_l=j$, $v_h\ne v_k$ for all $0\le h<k\le l$, and $(v_{h-1},v_h)\in\mc E$ for all $h=1,\ldots,l$. The gradient of a function $f:\R^n\to\R^m$ in a point $x$ of its domain is the matrix $\nabla f(x)\in\R^{m\times n}$ whose entries are the partial derivatives $[\nabla f(x)]_{ij}=\partial f_i(x)/\partial x_j$: if $m=h+l$ and the variable is explicitly written as $x=(y,z)$ where $y\in\R^h$ and $z\in\R^l$, then $\nabla_y f(y,z)\in\R^{h\times n}$ and $\nabla_z f(y,z)\in\R^{l\times n}$ are the left and right blocks of $\nabla f(x)\in\R^{m\times n}$, so that  $\nabla f(x)=[\nabla_y f(y,z),\nabla_z f(y,z)]$. Finally, we use the notation $a\vee b=\max\{a,b\}$ for the maximum between two scalar values $a,b\in\R$.

\section{Dynamical flow networks}\label{sec:linear0}

This paper focuses on single-commodity, first-order models of dynamical flow networks, also referred to as compartmental systems in some of the literature \cite{Jacquez.Simon:93,Walterand.Contreras:99}. These are dynamical systems with $n$-dimensional state vector $x=x(t)$ that belongs to the nonnegative orthant $\R_+^n$  at any time $t\ge0$.  The entries $x_i=x_i(t)$ of the state vector represent the mass in each \emph{cell} $i\in\mc I$, where $\mc I:=\{1,\ldots,n\}$ is a finite set of interconnected cells. 
The network flow dynamics can be compactly expressed as 
\be\label{compact}\dot x=u+F^T\1-F\1-w\ee
where: 
\begin{enumerate}
\item[(i)] $u\in\R_+^n$ is a nonnegative vector supported on a subset $\mc R\subseteq\mc I$ whose entries $u_i$ model the \emph{external inflows} in the cells $i\in\mc R$; 
\item[(ii)]$F\in\R_+^{n\times n}$ is a nonnegative matrix supported on a subset $\mc A\subseteq\mc I\times\mc I$ of ordered pairs of adjacent cells whose nonzero entries $F_{ij}$ represent the \emph{flow} from cell $i$ to cell $j$ for all pairs $(i,j)\in\mc A$;\footnote{Throughout, we will always assume that $(i,i)\notin\mc A$, so that $F_{ii}=0$, for all $i\in\mc I$.} 
\item[(iii)] and $w\in\R_+^n$ is a nonnegative vector supported on a subset $\mc S\subseteq\mc I$ whose entries $w_i$ model the \emph{outflows} from the cells $i\in\mc R$ towards the external environment. 
\end{enumerate}
Hence, in particular, we have 
\be\label{Fw}u_{i}=0\,,\quad i\notin\mc R\,,\qquad F_{ij}=0\,,\quad (i,j)\notin\mc A\,,\qquad w_i=0\,,\quad i\notin\mc S\,.\ee
While in typical applications the external inflows are ---either constant or time-varying--- exogenous inputs, the flow variables $F$ and $w$ in general may depend both on the state $x$ (thus allowing for feedback) and directly on the time $t$ (thus allowing for exogenous time variability). 

Entrywise rewriting of the dynamics \eqref{compact} reads
$$\label{entrywise}\dot x_i=u_i+\sum_{j\in\mc I}F_{ji}-\sum_{j\in\mc I}F_{ij}-w_i\,,\qquad i\in\mc I\,,$$
which is physically interpreted as a \emph{mass conservation law}: the rate of change of the mass in cell $i$ equals the imbalance between the total inflow in it and the total outflow from it, the former coinciding with the sum of the external inflow $u_i$ and the aggregate inflow from the other cells $\sum_{j\in\mc I}F_{ji}$, and the latter being given by the aggregate outflow towards other cells $\sum_{j\in\mc I}F_{ij}$ and the outflow towards the external environment $w_i$. 
Invariance of the nonnegative orthant $\R_+^n$ for the state vector $x$ is guaranteed by the additional constraint 
\be\label{eq:orthant-invariance}x_i=0\qquad\Longrightarrow\qquad w_i+\sum\nolimits_{j}F_{ij}=0\,,\ee
i.e., the outflow from an empty cell is always $0$. 

An equivalent form of \eqref{compact} that will often prove convenient is 
\be\label{eq:dynRz}\dot x=u-(I-R^T)z\,,\ee
where $$\label{z-def} z=F\1+w$$ is a nonegative $n$-dimensional vector whose entries $z_i=\sum_{j\in\mc I}F_{ij}+w_i$ represent the \emph{total outflows} from the cells $i\in\mc I$, and $R\in\R^{n\times n}$ is a \emph{routing matrix} whose entries $R_{ij}$, sometimes referred to as \emph{split ratios}, represent the fraction of outflow from cell $i$ directed towards cell $j$, i.e., they satisfy $$F_{ij}=R_{ij}z_i\,,\qquad  w_i=\left(1-\sum\nolimits_{j\in\mc I}R_{ij}\right)z_i\,,\qquad i,j\in\mc I\,.$$ The routing matrix $R$ is supported on $\mc A$, i.e., such that $R_{ij}=0$ for all $(i,j)\notin\mc A$ and it is substochastic, i.e., nonnegative and with row sums $\sum_jR_{ij}\le1$, with strict inequality possible only at the cells $i\in\mc S$ from which direct outflow towards the external environment is allowed. 

\begin{figure}
\begin{center}
\includegraphics[width=8.5cm]{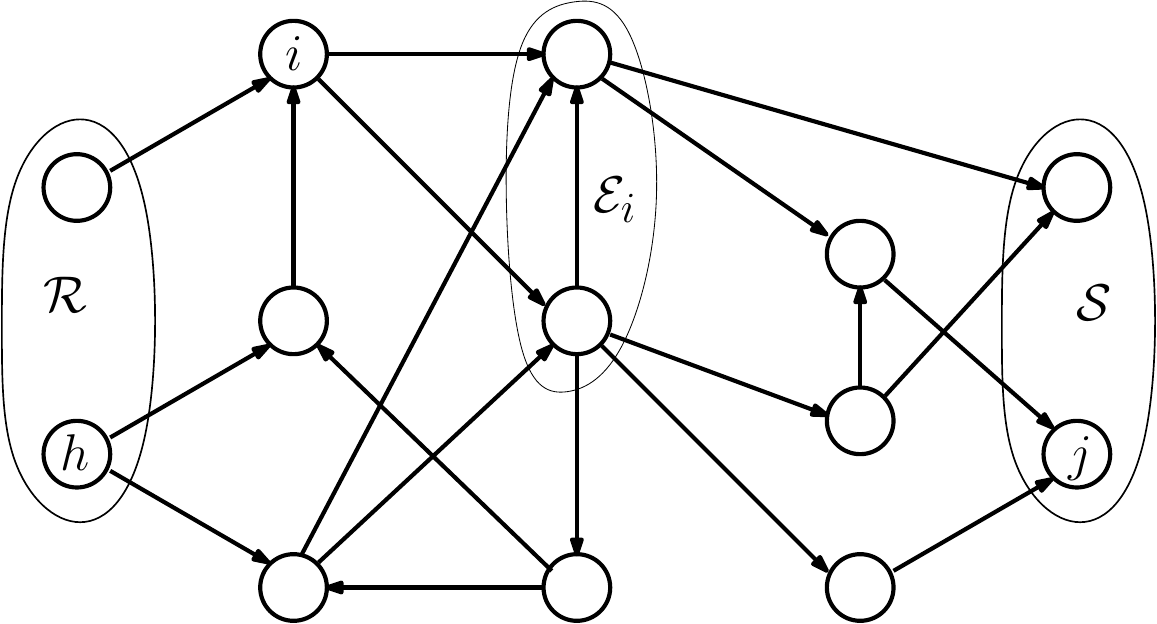}\end{center}
\caption{\label{fig:1} A network topology $\mc T=(\mc I,\mc A,\mc R,\mc S)$ whose nodes $i\in\mc I$ represent cells and links $(i,j)\in\mc A$ represent possible flows between adjacent cells. The subsets of cells $\mc R$ and $\mc S$ is where external inflow and, respectively, outflow towards the external environment, is possible. This network topology is both inflow-connected, as every cell $i\in\mc I$ is reachable from some $h\in\mc R$ by a directed path, and outflow-connected, as from every cell $i$ some cell $j\in\mc S$ can be reached by a directed path. The outneighborhood $\mc E_i$ of a cell $i$ is also highlighted. Note that, in this specific network topology, the out-neighborhoods of two nodes  either coincide or are disjoint. In fact, $\mc G=(\mc I,\mc A)$ is the line-digraph (see Remark \ref{ft}) of the digraph $\mc N=(\mc V,\mc I)$ displayed in Figure \ref{fig:2}.} 
\end{figure}

The sets $\mc I$, $\mc A$, $\mc R$, and $\mc S$ determine the  \emph{network topology}, to be denoted by $\mc T=(\mc I,\mc A,\mc R,\mc S)$. In particular, it is natural to interpret the sets $\mc I$ and $\mc A$, respectively, as the node set and  the link set of a directed graph $\mc G=(\mc I,\mc A)$: see Figure \ref{fig:1}. 
We refer to a cell $i\in\mc I$ as \emph{outflow-connected} if either $i\in\mc S$ or there exists a path in $\mc G$ starting in $i$ and ending in some $j\in\mc S$. The network topology $\mc T$ is referred to as outflow-connected if every cell $i\in\mc I$ is outflow-connected. Analogously,  a cell $i\in\mc I$ is referred to as \emph{inflow-connected} if either $i\in\mc R$ or there exists a path in $\mc G$ starting in some $j\in\mc R$ and ending in $i$, while the network topology is called inflow-connected if every cell $i$ is inflow-connected. 
Throughout the paper, we will use the notation $$\mc E_i:=\{k\in\mc I:\,(i,k)\in\mc A\}\,,\qquad i\in\mc I\,,$$ for the out-neighborhoods of the cells in $\mc T$. 

\begin{remark}\label{ft}In some applications, especially in road traffic network flows, the cells $i\in\mc I$ are often represented as links of a different digraph $\mc N=(\mc V,\mc I)$ whose node set $\mc V=\{v_0,v_1,\ldots,v_m\}$ includes a node $v_0$ representing the external environment and $m$ nodes $v_1,\ldots,v_m$ representing either junctions or interfaces between consecutive cells. (See Figure \ref{fig:2}.) In this case, the sets $$\mc R=\{i\in\mc I:\,i=(v_0,v_k)\text{ for some }1\le k\le m\}\,,\qquad \mc S=\{i\in\mc I:\,i=(v_k,v_0)\text{ for some }1\le k\le m\}$$ represent on-ramps and off-ramps, respectively, and $$\mc A=\{(i,j)\in\mc I\times\mc I:\,i=(v_{h},v_k)\text{ and }j=(v_{k},v_l)\text{ for some }0\le h,l\le m\text{ and }1\le k\le m\}$$ is the set of ordered pairs of adjacent cells. Clearly, to any such digraph $\mc N=(\mc V,\mc I)$, we can associate a network topology $\mc T=(\mc I,\mc A,\mc R,\mc S)$: in particular, $\mc G=(\mc I,\mc A)$ turns out to be the so-called \emph{line-digraph} \cite{Orlin:1978} of $\mc N=(\mc V,\mc I)$. Observe that the network topology $\mc T=(\mc I,\mc A,\mc R,\mc S)$ corresponding to the line-digraph of some $\mc N=(\mc V,\mc I)$ has the following properties: (a) cells $r\in\mc R$ are sources, i.e., they have no incoming links; (b) cells $s\in\mc S$ are sinks, i.e., they have no outgoing links; and (c) the out-neighborhoods $\mc E_i,\mc E_j$ of two cells $i,j\in\mc I$ either coincide (if $i$ and $j$ point towards the same junction $v\in\mc V\setminus\{v_0\}$) or they are disjoint (if they do not).\footnote{On the other hand, properties (a), (b), and (c) are not sufficient conditions for $\mc G=(\mc I,\mc A)$ to be the line-digraph of $\mc N=(\mc V,\mc I)$: necessary and sufficient conditions can be found in \cite{Orlin:1978}.} Hence, our choice of formulating results for the network topology $\mc T=(\mc I,\mc A,\mc R,\mc S)$ in which cells $i\in\mc I$ are represented as nodes is in the interest of both clarity and greater generality,  as the aforementioned conventional representation of road traffic networks can be recovered as a special case. \end{remark}  
\begin{figure}
\begin{center}
\includegraphics[width=7.5cm]{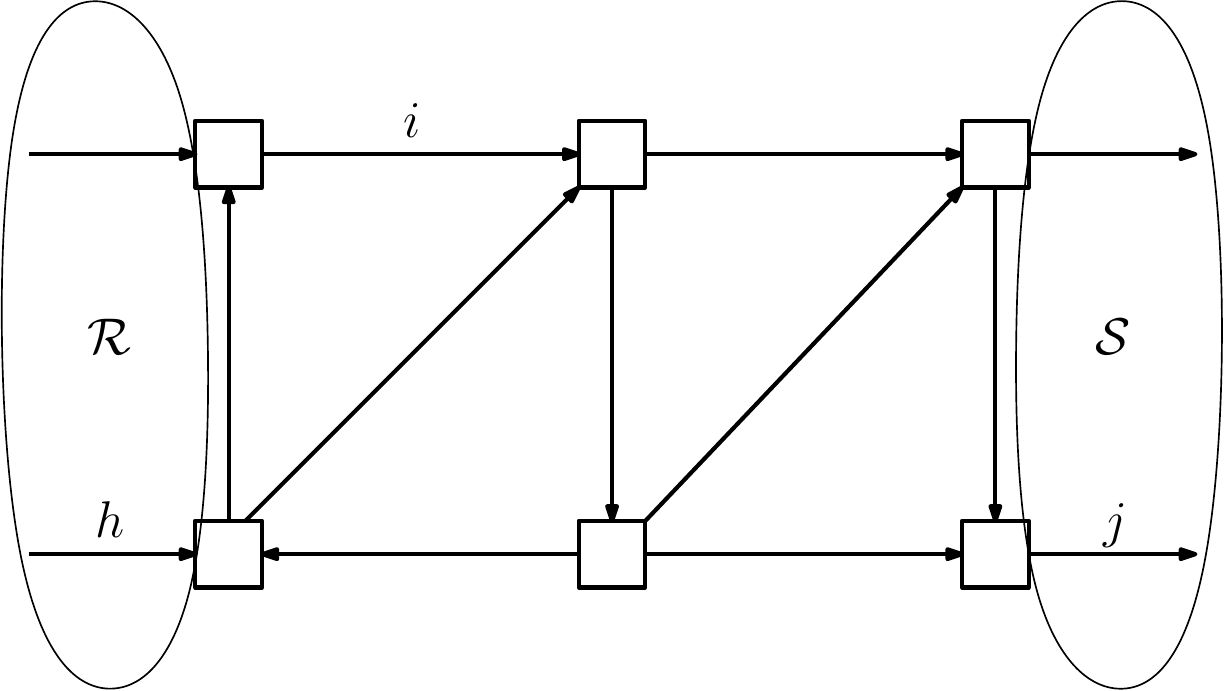} \end{center}
\caption{\label{fig:2} A digraph $\mc N=(\mc V,\mc I)$ representing a road traffic network: nodes $v\in\mc V$ stand for junctions and links $i\in\mc I$ represent cells that correspond to (portions of) roads. The subsets of cells $\mc R$ and $\mc S$ represent, respectively, the on-ramps and the off-ramps. The line digraph associated to $\mc N=(\mc V,\mc I)$ displayed above is the network topology $\mc T=(\mc I,\mc A,\mc R,\mc S)$ displayed in Figure \ref{fig:1}.} 
\end{figure}
\medskip



%
%
%
%
%
%

\section{Affine Dynamical Flow Networks} \label{sec:linear}
The simplest special case of dynamical flow networks is the class of \emph{affine dynamical flow networks}, characterized by the property that the flows $F_{ij}$ and $w_i$ are linear functions of the state vector $x$. In fact, nonnegativity of the flow variables and dynamic invariance of the nonnegative orthant $\R_+^n$ imply that such linear relationships necessarily take the form 
$$F_{ij}=\alpha_{ij}x_i\,,\qquad w_i=\beta_{i}x_i\,,\qquad i,j\in\mc I\,,$$
where the $\alpha\in\R_+^{n\times n}$ and $\beta\in\R_+^{n}$ are such that $\alpha_{ij}>0$ for all pairs of adjacent cells $(i,j)\in\mc A$, $\alpha_{ij}=0$ for all $i\ne j$ such that $(i,j)\notin\mc A$, while $\beta_i>0$ for $i\in\mc S$ and $\beta_i=0$ for $i\in\mc I\setminus\mc S$. By putting $$L_{ij}=-\alpha_{ij}\,,\qquad L_{ii}=\beta_i+\sum_{j\ne i}\alpha_{ij}\,,\qquad i\ne j\in\mc I\,,$$ and assembling the coefficients $L_{ij}$ in an $n\times n$ matrix $L$, one gets the following form for affine dynamical flow networks 
\be\label{eq:affine1}\dot x=u-L^Tx\,.\ee
Observe that $-L_{ij}=\alpha_{ij}\ge0$ for all $i\ne j\in\mc I$, i.e., $-L$ is Metzler, and $ L\1=\beta\ge0$. Hence, $-L$ is a compartmental matrix and, in particular, equation \eqref{eq:affine1} describes a positive affine system \cite{Farina.RInaldi:2000}.  We refer  to the matrix $L$ as outflow-connected if the corresponding network topology $\mc T$ is. Observe that the diagonal $d=(L_{ii})_{i\in\mc I}$ of an outflow-connected compartmental matrix $L$ is strictly positive and, upon defining the substochastic routing matrix $R=I-D^{-1}L$, where $D=\diag(d)$, equation \eqref{eq:affine1} can be equivalently rewritten as 
\be\label{eq:affine2}\dot x=u-(I-R^T)Dx\,.\ee

Some stability properties of affine dynamical flow networks are gathered in the following proposition. 

\begin{proposition}\label{prop:1} Every affine dynamical flow network \eqref{eq:affine1} is non-expansive in the $l_1$-distance, i.e., the $l_1$-distance between any two of its solutions is never increasing in time.
Moreover, if the compartmental matrix $-L$ is outflow-connected, then it is also Hurwitz stable, and, furthermore, if the inflow vector $u$ is constant, then $x^*=(L^T)^{-1}u$ is a globally exponentially stable equilibrium of \eqref{eq:affine1} and $V(x)=\sum_{i\in\mc I}|x_i-x^*_i|$ is a Lyapunov function. 
\end{proposition}
\begin{proof}
For any two solutions $x(t)$ and $\tilde x(t)$ of \eqref{eq:affine1}, let $y(t)=x(t)-\tilde x(t)$ be their difference, and note that $\dot y=-L^Ty$. For all $t,s\ge0$, one has that 
\be\label{l1}||x(t+s)-\tilde x(t+s)||_1-||x(t)-\tilde x(t)||_1=||y(t+s)||_1-||y(t)||_1=\int_t^{t+s}-y(\tau)^TL\sgn(y(\tau))\de \tau\,,\ee where, $\sgn(y)\in\{-1,0,1\}^n$ stands for the sign vector of $y$. Now, for any $y\in\R^n$, one can partition the cell set as $\mc I=\mc I_{-}\cup\mc I_0\cup\mc I_{+}$ according to the sign of $y$'s entries. Let $L_{(s,r)}$ and $y_{(s)}$ for $s,r\in\{-,0,+\}$ be the corresponding blocks of $L$ and $y$, respectively, and notice that 
\be\label{sum}y^TL\sgn(y)=y_{(+)}^TL_{(+,+)}\1-y_{(+)}^TL_{(+,-)}\1+y_{(-)}^TL_{(-,+)}\1-y_{(-)}^TL_{(-,-)}\1\,.\ee
Now, note that, since $-L$ is Metzler, the blocks $L_{(+,-)}$ and $L_{(+,-)}$ are nonpositive, so that $y_{(+)}^TL_{(+,-)}\1\le0$ and $y_{(-)}^TL_{(-,+)}\1\ge0$. On the ther hand, since $-L$ is Metzler and $L\1\ge0$, one has that $L_{(+,+)}\1\ge0$ and $L_{(-,-)}\1\ge0$. In turn, these inequalities imply, respectively, that $y_{(+)}^TL_{(+,+)}\1\ge0$ and $y_{(-)}^TL_{(-,-)}\1\le0$. Hence, all the addends in the right-hand side of \eqref{sum} are nonnegative, so that $$y^TL\sgn(y)\ge0$$ and \eqref{sum} implies that $$||x(t+s)-\tilde x(t+s)||_1\le||x(t)-\tilde x(t)||_1\,,$$ so that the first claim is proved. 

That Hurwitz stability of a compartmental matrix $-L$ is equivalent to its outflow-connectivity is a classical result following from the Perron-Frobenius theory, see, e.g., \cite[Theorem 3]{Jacquez.Simon:93}. Since the dynamical system \eqref{eq:affine1} is affine, Hurwitz stability of $-L$ is equivalent to global exponential stability of the equilibrium $x^*$. Finally, that the $l_1$-distance from equilibrium $||x-x^*||_1$ is a Lyapunov function follows from the first claim by letting one of the two solutions start from the equilibrium $x^*$ itself. \qed\end{proof}\medskip

\begin{remark}\label{remark-equilirbiumpositivity}
Observe that, using the decomposition $L=D(I-R)$ the equilibrium point  $x^*=(L^T)^{-1}u$ of an outflow-connected affine dynamical flow network \eqref{eq:affine1} may be rewritten as
\be\label{x*expansion}x^*=D^{-1}z^*\,,\qquad z^*=(I-R^T)^{-1}u=u+R^Tu+(R^{2})^Tu+(R^{3})^Tu+\ldots\,.\ee
The rightmost side of equation \eqref{x*expansion} states that the entries of the equilibrium point satisfy $x^*_i=z_i^*/L_{ii}$ where $z_i^*$ is the aggregate of the external inflow $u_i$ in cell $i$, plus the weighted sum $\sum_jR_{ji}u_j$ of the external inflows in the cells $j$ immediately upstream of $i$, plus the weighted sum $\sum_j(R^2)_{ji}u_j$ of the external inflows in the cells $j$ two hops upstream of $i$, and so on. Here, convergence of the nonnegative term series $\sum_{k\ge0}(R^T)^ku$ is guaranteed by the fact that the spectral radius of $R$ is strictly less than $1$. Recall that, for every inflow-connected cell $i\in\mc I$ there exists some cell $h\in\mc R$ from which $i$ is reachable by a directed path (c.f.~Figure \ref{fig:1}) in the digraph $\mc G=(\mc I,\mc A)$, so that there exists some $k\ge0$ such that $(R^k)_{hi}u_h>0$. It then follows from \eqref{x*expansion} that, if the network is both inflow- and outflow-connected, then the equilibrium $x^*$ has strictly positive entries. \end{remark}

\begin{remark}\label{remark-l1}While never increasing, the $l_1$-distance from equilibrium $\sum_{i\in\mc I}|x_i-x^*_i|$ of the state vector of a general outflow-connected affine dynamical flow network, does not necessarily decay to $0$ exponentially fast as time grows large (it does so in the special case when $\mc S=\mc I$). Exponential stability of \eqref{eq:affine1} can instead be certified by weighted $l_1$-distances of the form $\sum_i\omega_i|x_i-x_i^*|$, for positive values of the $\omega_i$'s that depend on the specific values of the entries of $L$. In fact, existence of such sum- as well as max-separable Lyapunov functions is equivalent to stability for positive (not only compartmental) systems \cite[Proposition 1]{Rantzer:2015}, and this is at the heart of the scalability properties of these systems. (See \cite{Coogan:16} for extensions of this result to nonlinear monotone systems.) In contrast, the unweighted $l_1$-distance can be used to prove global asymptotic stability of compartmental systems (via LaSalle's theorem) using structural properties only, i.e., properties that are independent of the specific values of the entries of the compartmental matrix $-L$, but depend just on their sign pattern. As we shall see in the Section \ref{sec:nonlinear-stability}, such structural properties of the $l_1$-distance carry over to a signifcant class of nonlinear dynamical flow networks for which they are very useful in establishing stability and robustness results.\end{remark}

We conclude this section by noticing that affine dynamical flow network models cannot account for congestion effects as they prescribe that the total outflow from a cell $z_i=L_{ii}x_i$ is a linear function of the mass $x_i$ on it. Stated in other terms, in such affine dynamical flow network models, the rate $z_i/x_i$ at which single units of mass flow out of a cell $i$ is a constant $L_{ii}$ that does not decrease when the  mass $x_i$ in that cell, or the mass $x_j$ in other neighbor cells $j$ increase. 

%
%
%

\section{Monotone dynamical flow networks and their stability} \label{sec:nonlinear-stability}
In this and the next sections, we study a class of nonlinear dynamical flow networks and analyze their stability and robustness properties. 
While it is known that nonlinear dynamical flow networks can exhibit arbitrarily complex (including chaotic) dynamic behaviors \cite{Jacquez.Simon:93}, we focus on a class of nonlinear dynamical flow networks which is rich enough to encompass many relevant examples and at the same time has enough additional system structure to allow for tractable analysis and synthesis. Such dynamical system structure proving particularly relevant in this context is \emph{monotonicity}. 

Recall that a dynamical system is referred to as monotone if its trajectories preserve the partial order between initial states and inputs. \cite{Hirsch:1983a,Smith:1995,HirschMonotoneSystems06}
A classical result known as Kamke's theorem \cite{KamkeAM32,AngeliTAC03} establishes that, for dynamical systems of the form $\dot x=f(x,u)$ with Lispchitz-continuous $f$, monotonicity is equivalent to $\nabla_uf(x,u)$ being a nonnegative matrix and $\nabla_xf(x,u)$ being a Metzler matrix for almost every $x$ and $u$.  Considering their structure \eqref{compact}, we propose the following definition of monotone dynamical flow networks. 
\begin{definition}\label{def:monotone}
A dynamical flow network $$\dot x=u+f(x)\,,\qquad f(x)=F^T(x)\1-F(x)\1-w(x)\,,$$ 
with Lipschitz-continuous flow functions $F(x)$ and $w(x)$, is \emph{monotone} in a domain $\mc D\subseteq\R^n_+$ if $(\nabla f(x))^T$ is a compartmental matrix, i.e., if
\be\sum_{k\in\mc I}\left(\frac{\partial }{\partial x_j}F_{ki}(x)-\frac{\partial }{\partial x_j}F_{ik}(x)\right)-\frac{\partial }{\partial x_j}w_i(x)\ge0\,,\qquad \sum_{k\in\mc I}\frac{\partial}{\partial x_j}w_k(x)\ge0\,,\qquad \forall i\ne j\in\mc I\,,\label{eq:monotone}\ee
for almost every $x\in\mc D$.
\end{definition}

Observe that equation \eqref{eq:monotone} requires that the transpose of the Jacobian matrix of $f(x)$ be compartmental rather than simply Metzler, hence it is more restrictive than Kamke's conditions for monotone dynamical systems, also referred to as cooperative systems \cite{HirschJMA85,HirschPS03}. In fact, we have already pointed out in Section \ref{sec:linear} that the matrix $-L^T$ determining the dynamics of affine dynamical flow networks \eqref{eq:affine1} is a compartmental matrix. I.e., affine dynamical flow networks are always monotone:  
 this can be traced back to the fact that, for linear systems, preservation of the nonnegative orthant is equivalent to preservation of the partial ordering between trajectories. 
 On the other hand, for nonlinear dynamical flow networks, monotonicity is an additional property that further constraints their dynamical behavior. However, as it turns out, many nonlinear dynamical flow networks of interest are monotone either globally (i.e., on the whole state space) or locally (i.e., in a subdomain of their state space typically including the origin $0$). 
 
 We now present a first example of nonlinear dynamical flow network here, while five more examples will be discussed in Section \ref{sec:demand-supply}. 
 
 \begin{example}\label{ex6}
For a network topology $\mc T=(\mc I,\mc A,\mc R,\mc S)$ that is both inflow- and outlflow-connected, and a constant inflow vector $u$ supported on $\mc R$, consider the following static network flow optimization problem \cite{Bertsekas:1998,Whittle:2007}: 
\be\label{conv-min}\min_{\substack{\\F\in\R_+^{n\times n},w\in\R_+^{n}:\\\eqref{Fw}\text{ and}\\[1pt]u+F^T\1=F\1+w}}\quad
\sum_{(i,j)\in\mc A}\psi_{ij}(F_{ij})+\sum_{k\in\mc S}\psi_k(w_k)\,,\ee
where $\psi_{ij}(\,\cdot\,)$, for $(i,j)\in\mc A$, and $\psi_k(\,\cdot\,)$, for $k\in\mc S$, are nonnegative-real-valued, twice differentiable, strictly convex, and increasing cost functions, such that $$\lim\limits_{y\to+\infty}\psi'_{ij}(y)=\lim\limits_{y\to+\infty}\psi'_k(y)=+\infty\,,\qquad (i,j)\in\mc A\,,\qquad k\in\mc S\,.$$  
The Lagrangian of the minimization problem \eqref{conv-min} reads 
$$\label{Lagrangian}
\ba{rcl}L(F,w,x)&=&
\ds\sum_{(i,j)\in\mc A}\psi_{ij}(F_{ij})+\sum_{k\in\mc S}\psi_k(w_k)+\sum_{i\in\mc I}x_i\Big(u_i+\sum_{j\in\mc I}F_{ji}-\sum_{j\in\mc I}F_{ij}-w_i\Big)\\[15pt] 
&=&\ds\sum_{(i,j)\in\mc A}\left(\psi_{ij}(F_{ij})-(x_i-x_j)F_{ij}\right)+\sum_{k\in\mc S}\left(\psi_k(w_k)-w_kx_k\right)+\sum_{i\in\mc R}x_iu_i\,,
\ea$$
where $x_i$ is the Lagrange multiplier associated to the mass conservation constraint at cell $i$. 
We then get the first-order conditions for optimality 
\be\label{optFw}F_{ij}=\left\{\ba{lcl}0&\se&x_i-x_j<\psi_{ij}'(0)\\[5pt]
(\psi'_{ij})^{-1}(x_i-x_j)&\se&x_i-x_j\ge\psi_{ij}'(0)\,,\ea\right.\qquad 
w_{k}=\left\{\ba{lcl}0&\se&x_k<\psi_{k}'(0)\\[5pt]
(\psi'_{k})^{-1}(x_k)&\se&x_k\ge\psi_{k}'(0)\,.\ea\right.\ee
Observe that strict convexity of the cost functions guarantees invertibility of their derivatives in the range $[\psi_{ij}'(0),+\infty)$ and $[\psi_{k}'(0),+\infty)$, for all $(i,j)\in\mc A$ and $k\in\mc S$. The continuous-time dual ascent dynamics for the optimization problem \eqref{conv-min} is then given by \eqref{compact} where the flow variables $F=F(x)$ and $w=w(x)$ are given by \eqref{Fw} and \eqref{optFw}. Observe that, if $x_i=0$ and $x_j\ge0$ for some $(i,j)\in\mc A$, then $x_i-x_j\le0\le\psi_{ij}'(0)$ so that $F_{ij}=0$. Similarly, $x_k=0$ for $k\in\mc S$ implies that $x_k\le\psi_{k}'(0)$, so that $w_k=0$. Hence, the sufficient condition \eqref{eq:orthant-invariance} for invariance of nonnegative orthant $\R_+^n$ is satisfied, so that \eqref{compact}, \eqref{Fw}, and \eqref{optFw} actually define a dynamical flow network. To verify that it is a monotone dynamical flow network, simply observe that $$\frac{\partial}{\partial x_j}F_{ji}\ge0\,,\qquad\frac{\partial}{\partial x_j}F_{ij}\le0\,,\qquad\frac{\partial}{\partial x_j}F_{ik}=\frac{\partial}{\partial x_j}F_{ki}=0\,,\qquad\frac{\partial}{\partial x_j}w_j\ge0\,, \qquad\frac{\partial}{\partial x_j}w_k=0\,,$$ for all choices of three different cells $i,j,k\in\mc I$. Hence, for $f(x)=F^T(x)\1-F(x)\1-w(x)$ $$\frac{\partial}{\partial x_j}f_{i}(x)=\frac{\partial}{\partial x_j}F_{ji}-\frac{\partial}{\partial x_j}F_{ij}\ge0\,,\qquad\sum_k\frac{\partial}{\partial x_j}w_{k}=\frac{\partial}{\partial x_j}w_{j}\ge0\,,$$ so that \eqref{eq:monotone} is satisfied. Therefore, the dual ascent dynamics defined by \eqref{compact}, \eqref{Fw}, and \eqref{optFw} is a monotone dynamical flow network on its whole state space $\R_+^n$. 
\end{example}

In the following theorem we gather a few stability properties that are enjoyed by monotone dynamical flow networks. For the sake of simplicity, we state the results for dynamical flow networks that are monotone on the whole state space, whereas they can be adapted to the cases where the monotonicity property is satisfied in a subdomain of the state space only. 

\begin{theorem}\label{theo:stability-monotone}
For a monotone dynamical flow network $$\dot x=u+F^T(x)\1-F(x)\1-w(x)\,,$$ 
\begin{enumerate}
\item[(i)] the partial order between external inflows and initial states is preserved by the state at any time; 
\item[(ii)] the $l_1$-distance between any two solutions is never increasing.  
\end{enumerate} 
Moreover, if the external inflows $u$ are constant in time, then: 
\begin{enumerate}
\item[(iii)] every equilibrium $x^*$ is stable; 
\item[(iv)] an equilibrium $x^*$ is globally asymptotically stable if and only if it is locally asymptotically stable; 
\item[(v)] an equilibrium $x^*$ is asymptotically stable if the compartmental matrix $(\nabla f(x^*))^T$ is outflow-connected; 
\item[(vi)] the trajectory started from $x(0)=0$ is always entrywise monotonically nondecreasing in time, hence convergent to a (possibly infinite) limit. 
\end{enumerate} 
\end{theorem}
\begin{proof} Claim (i) follows from Kamke's theorem \cite{KamkeAM32} and its generalization to systems with input \cite{AngeliTAC03}. 
Claim (ii) is a generalization of the first claim of Proposition \ref{prop:1} and follows from \cite[Lemma 1]{Como.Lovisari.ea:TCNS15}. Claim (iii) is an immediate consequence of (ii): when one of the two trajectories is an equilibrium $x^*$, the $l_1$-distance $||x(t)-x^*||_1$ is non-increasing in time, so that $x^*$ is necessarily stable.  Claim (iii) is also a consequence of (ii) as proved in \cite[Lemma 6]{Lovisari.Como.ea:CDC14}. Claim (v) is a consequence of the Lyapunov linearization method and Proposition \ref{prop:1}. Finally, to prove claim (vi), let $\Phi^{(t)}(x)$ denote the associated semi-flow and let $x(t)$ be the solution started with $x(0)=0$: then, $x(t+s)=\Phi^{(t)}(x(s))\ge\Phi^{(t)}(0)=\Phi^{(t)}(x(0))=x(t)$, where the inequality follows from the fact that $x(s)\ge0$ (since the nonnegative orthant is dynamically invariant) and system monotonicity. \qed
\end{proof}\medskip

Quite remarkably, Theorem \ref{theo:stability-monotone} relies solely on structural properties of the dynamical system, as it applies to monotone dynamical flow networks irrespective of their specific form. In particular, the non-expansiveness of the $l_1$-distance is an extension of the analogous result for affine dynamical flow networks (Proposition \ref{prop:1}). The related issue of existence of sum-separable Lyapunov functions for stable monotone systems has recently attracted significant attention \cite{Dirr.ea:2015}. Non-expansiveness  and contractivity of properly selected metrics, also known as differential stability, in dynamical systems is an even stronger property than Lyapunov stability and has been investigated, e.g., in \cite{Lohmiller.Slotine:1998,Sontag:10}. Finally, monotonicity and $l_1$-contraction for dynamical flow networks have analogies with properties of some PDE models of dynamical flows, see, e.g., Kruzkov's Theorem \cite[Proposition 2.3.6]{Serre:99} for entropy solutions of scalar conservation laws. 

As a first illustration of the applicability range of Theorem \ref{theo:stability-monotone}, we prove that the dual ascent dynamics introduced in Example \ref{ex6} have a globally asymptotically stable equilibrium that corresponds to a minimizer  of the convex network flow optimization \eqref{conv-min}. 
\addtocounter{example}{-1}
\begin{example}\textbf{(continued)}
Consider the dual ascent dynamics defined by \eqref{compact}, \eqref{Fw}, \eqref{optFw}, with initial state $x(0)=0$. Since this is a monotone dynamical flow network, by point (vi) of Theorem \ref{theo:stability-monotone}, $x_i(t)$ converges to a possibly infinite value $x_i^*$, for all cells $i\in\mc I$. We now prove that $x^*_i<+\infty$ for all $i\in\mc I$. To see this, let $$\mc J=\{i\in\mc I:\,x_i^*=+\infty\}\,,\qquad\mc K=\{i\in\mc I:\,x_i^*<+\infty\}\,,$$ 
and observe that the limit flows $$F^*=\lim\limits_{x\to x^*}F\,,\qquad w^*=\lim\limits_{x\to x^*}w$$ satisfy
\be\label{F*}F_{jk}^*=+\infty\,,\quad j\in\mc J,k\in\mc K:(j,k)\in\mc A\,,\qquad 
F_{kj}^*=0\,,\quad j\in\mc J,k\in\mc K\,,\qquad 
w_h^*=+\infty \,,\quad h\in\mc S\cap\mc J\,.\ee
Notice that, if $\mc J$ contains at least one outflow-connected cell, then either $\mc J\cap\mc S$ is nonempty, or there must exist some $j\in\mc J$ and $k\in\mc K$ such that $(j,k)\in\mc A$. It then follows from \eqref{F*} that
$$\lim_{x\to x^*}\sum_{j\in\mc J}\dot x_j=\sum_{j\in\mc J}u_j+\sum_{j\in\mc J}\sum_{k\in\mc K}F^*_{kj}-\sum_{j\in\mc J}\sum_{k\in\mc K}F^*_{jk}-\sum_{j\in\mc J}w_j^*=-\infty\,,$$ 
which contradicts the fact $\lim_{x\to x^*}\sum_{j\in\mc J}x_j=+\infty$. Since the network is outflow-connected, then necessarily $\mc J$ is empty so that $x^*$ is finite, and hence $x^*$ is an equilibrium of \eqref{compact}, \eqref{Fw}, \eqref{optFw}. Hence $x^*$, $F^*$, and $w^*$ satisfy the mass conservation constraint $u+(F^*)^T\1=F^*\1+w^*$ as well as \eqref{optFw}, so that $F^*$, and $w^*$ are minimizers in the static network flow optimization \eqref{conv-min}. Observe that the equilibrium $x^*$ is necessarily stable by point (iii) of Theorem \ref{theo:stability-monotone}. To prove its global 
asymptotic stability one can build on the strict convexity of the cost functions and inflow- and outflow-connectivity of the network topology in order to prove that the compartmental matrix $(\nabla f(x^*))^T$ is outflow-connected and then use points (v) and (iii) of Theorem \ref{theo:stability-monotone}. We observe that dynamical flow network structures similar to the ones considered in this example can be found in the context of dissipative flow networks \cite{Zlotnik.ea:2015}.
\end{example}
The arguments developed in Example \ref{ex6} can be extended to cover the case of finite capacities constraints $F_{ij}<c_{ij}$ on the cell-to-cell flows for  $(i,j)\in\mc A$ and  $w_k<c_k$ on the external outflows for $k\in\mc S$, by allowing the considered convex costs to be extended-real-valued with $\psi_{ij}(y)=+\infty$ if and only if $y\ge c_{ij}$ and $\psi_{k}(y)=+\infty$ if and only if $y\ge c_{k}$. We do not include details on such extention here, as we will deal with finite capacities bounds $z_i\le C_i$ on the total outflows from the different cells within a different setting  in the next sections. 

\section{Nonlinear dynamical flow networks with demand and supply constraints}\label{sec:demand-supply}

Whilst linear relations between the flow variables $F$ and $w$ and the state $x$ studied in Section \ref{sec:linear} fail to capture congestion effects in dynamical flow networks, such effects can be modeled by nonlinearities. A class of such nonlinearities has already been considered in Example \ref{ex6} of Section \ref{sec:nonlinear-stability} within the framework of dual ascent algorithms in network flow optimization. In this section, we introduce another class of such nonlinear dynamical flow networks whereby the outflows from, and the inflows in, the cells are bounded by demand and, respectively, supply functions of their current state. Such class of models is particularly relevant for road traffic flow networks (where the demand and supply functions represent, respectively, the rising and decreasing parts of the fundamental diagram \cite{LighthillTrafficPTRS55,Richards:56}), as well as production networks (where demand functions represent clearing functions \cite{Karmarkar:89}), and data networks (where demand functions can be understood as representing communication link throughput). We present five examples of these models, show that they all fit within the framework of monotone dynamical flow networks either globally, or in a relevant subdomain of their state space, and finally use Theorem \ref{theo:stability-monotone} to establish their stability properties.

\begin{figure}
\begin{center}
\includegraphics[width=7.5cm]{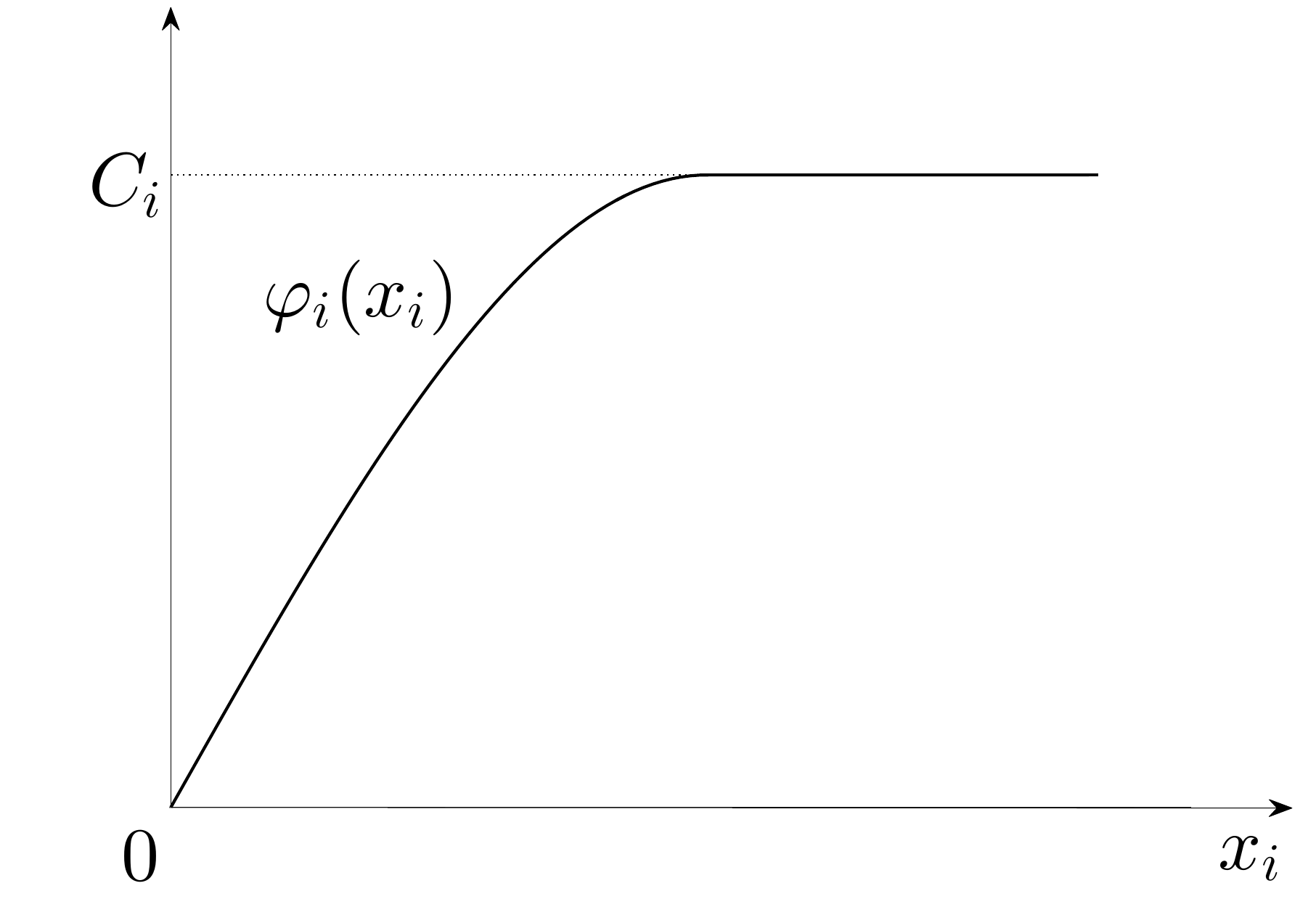}
\hspace{1cm}
\includegraphics[width=7.5cm]{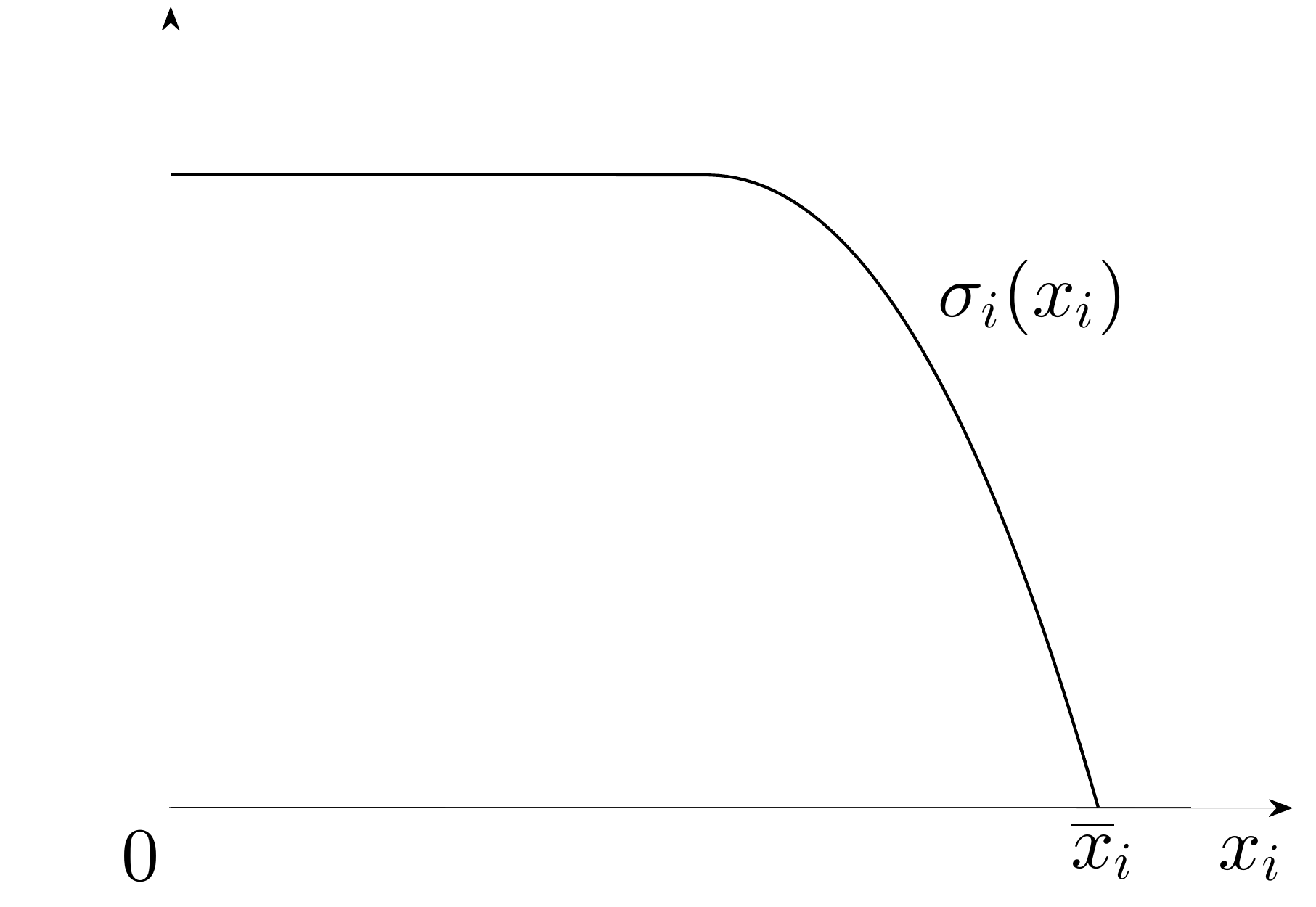}
 \end{center}
\caption{\label{fig:demand-supply} A demand function (on the left) and a supply function (on the right).} 
\end{figure}

We start by introducing \emph{demand functions} $\varphi_i(x_i)$ that characterize the maximum possible outflow $z_i$ from the cells $i\in\mc I$ given the current mass $x_i$, so that $0\le z_i\le\varphi_i(x_i)$. Throughout this paper, such demand functions $\varphi_i(x_i)$ are assumed to be Lipschitz-continuous, non-decreasing (the more mass, the more demand of outflow from the cell), concave (the more mass the less outflow increment), with $\varphi_i(0)=0$ (no outflow from cell $i$), and such that the cell \emph{flow capacity} $$C_i:=\sup_{x_i\ge0}\varphi_i(x_i)$$ is strictly positive and finite for every cell $i$ (see the leftmost plot in Figure \ref{fig:demand-supply}). We stack up the demand functions in an $n$-dimensional vector $\varphi(x)=(\varphi_i(x_i))_{i\in\mc I}$ and compactly write the demand constraints as 
\be\label{eq:demand-constraint}0\le z\le \varphi(x)\,.\ee

We may equivalently rewrite the demand constraint \eqref{eq:demand-constraint} by introducing a vector $\gamma\in\R^n$ whose entries $\gamma_i\in[0,1]$ for $i\in\mc I$ are flow control parameters, and put $$z=\Gamma\varphi(x)\,,\qquad\Gamma=\diag(\gamma)\,.$$ With this notation, equation \eqref{eq:dynRz} for nonlinear dynamical flow networks with demand constraints reads 
\be\label{eq:RGamma}\dot x=u-(I-R^T)\Gamma\varphi(x)\,.\ee
In the above, the external inflow vector $u$, the substochastic routing matrix $R$, and the flow control parameters matrix $\Gamma$ might be either constant or time-varying, in the latter case depending on either exogenous factors or feedback information. Recall that the routing matrix is always substochastic and needs to satisfy the network topology constraints $R_{ij}=0$ for all $(i,j)\notin\mc A$ and $\sum_{j\in\mc I}R_{ij}=1$ for all cells $i\notin\mc S$. Observe that invariance of the nonnegative orthant for the dynamics \eqref{eq:RGamma} is implied by the demand constraint \eqref{eq:demand-constraint} and the fact that $\varphi_i(0)=0$ so that there can be no outflow from an empty cell.  We gather below a few examples discussed in the literature that fit in this framework.
 
\begin{example}\label{ex1}
The simplest instance of a nonlinear dynamical flow network with demand constraints is in the form 
\be\label{eq:R}\dot x=u-(I-R^T)\varphi(x)\,,\ee with constant external inflow vector $u\in\R_+^n$ and constant routing matrix $R\in\R_+^{n\times n}$. Here, the demand constraints \eqref{eq:demand-constraint} are met with equality, i.e., $z=\varphi(x)$, so that $\Gamma=I$. The 
network topology is $\mc T=(\mc I,\mc A,\mc R,\mc S)$, where $$\mc I=\{1,\ldots,n\}\,,\qquad \mc A=\{(i,j):\,R_{ij}>0\}\,,\qquad \mc R=\{i\in\mc I:\,u_i>0\}\,,\qquad\mc S=\left\{i\in\mc I:\,\sum\nolimits_jR_{ij}<1\right\}\,.$$ Provided that $\mc T$ is outflow-connected, it can be easily verified that a necessary and sufficient condition for the existence of an equilibrium of the dynamics \eqref{eq:R}
is that the vector $$z^*=(I-R^T)^{-1}u$$ satisfies \be\label{eq:z<C}z_i^*<C_i\,,\qquad i\in\mc I\,.\ee In fact, if this condition is met, and the demand functions are strictly increasing, then there exists a unique equilibrium \be\label{eq:x*}x^*=\varphi^{-1}(z^*)\,.\ee 
On the other hand, it is not hard to verify that \eqref{eq:R} is a monotone dynamical flow network, since 
$$\frac{\partial }{\partial x_k}F_{ki}(x)=R_{ki}\varphi'_k(x_k)\ge0\,,\qquad \frac{\partial }{\partial x_j}F_{ik}(x)=0\,,\qquad 
i,j,k\in\mc I\,,\qquad j\ne i$$
$$ \frac{\partial }{\partial x_j}w_i(x)=0\,,\qquad \frac{\partial }{\partial x_i}w_i(x)=\left(1-\sum\nolimits_kR_{ik}\right)\varphi'_i(x_i)\ge0\,,\qquad 
i,j\in\mc I\,,\qquad j\ne i\,.$$
Therefore, when \eqref{eq:z<C} holds true, the equilibrium \eqref{eq:x*} is stable by point (iii) of Theorem \ref{theo:stability-monotone}, while its 
global asymptotic stability then follows from points (iv) and (v) of Theorem \ref{theo:stability-monotone}.
\end{example}
\begin{example}\label{ex2} In the works \cite{ComoPartITAC13,ComoPartIITAC13}, nonlinear dynamical flow networks are considered of the form 
\be\label{eq:Rx}\dot x=u-(I-R^T(x))\varphi(x)\,,\ee
with constant external inflow vector $u$ and increasing demand functions $\varphi(x)$. As in \eqref{eq:R} of Example \ref{ex1}, the demand constraints met with equality, i.e., $z=\varphi(x)$. The main difference with respect to \eqref{eq:R} is the \emph{locally responsive feedback} routing matrix $R(x)$. By locally responsive feedback it is meant that: 
\begin{enumerate}
\item[(a)] the split ratios $R_{ij}(x)$ depend on local information only, consisting of the mass $x_k$ in the cells $k\in\mc E_i$, where $\mc E_i$ is the outneighborhood of cell $i$ in the network topology; 
\item[(b)] the dependence on such local information is monotone in the sense that 
\be\frac{\partial }{\partial x_k}\sum_{h\in\mc E_i}R_{ih}(x)\le0\,,\qquad\frac{\partial }{\partial x_k}R_{ij}(x)\ge0\,,\qquad i\in\mc I\,,\qquad j,k\in\mc E_i\,,\quad  j\ne k\,;\ee
\item[(c)] if $x_j\to+\infty$ for some but not all $j\in\mc E_i$, or if $i\in\mc S$ and $x_j\to+\infty$ for all $j\in\mc E_i$, then
\be\lim R_{ij}(x)=0\,.\ee
\end{enumerate}
Property (a) above guarantees that the routing policy $R(x)$ is decentralized, in the sense that it depends on local information only. This in particular implies that the network flow dynamics \eqref{eq:Rx} is distributed as the righthand side of $$\dot x_i=u_i+\sum\nolimits_{j}R_{ji}(x)\varphi_j(x_j)-\varphi_i(x_i)$$ depends only on the external inflow $u_i$ and the states of cell $i$, of its immediately upstream cells (i.e., such $j$ such that $i\in\mc E_j$), as well as their immediately downstream cells (i.e., those cells $k$ such that $i,k\in\mc E_j$ for the same $j$). Property (b) is a relatively natural monotonicity condition: it states that, if the mass in an out-neighbor cell $k\in\mc E_i$ increases, then the fraction $R_{ij}(x)$ of outflow from cell $i$ that is routed towards any other out-neighbor cell $j\in\mc E_i\setminus\{k\}$  cannot decrease and, if $i\in\mc S$, the same holds true for the fraction $1-\sum_{h\in\mc E_i}R_{ih}(x)$  that is routed towards the external environment directly: as a consequence, this implies that the fraction routed towards cell $k$ itself cannot increase. Finally, property (c) ensures that, if the mass $x_j$ in a cell $j\in\mc E_i$ grows large, and either $i\in\mc S$ (so that flow can be routed directly to the external environment) or there are other cells $k\in\mc E_i$ where the mass stays bounded, then the fraction $R_{ij}(x)$ of flow routed from $i$ to $j$ vanishes. 
Specific forms of such locally responsive feedback routing policies include the $i$-logit function
\be\label{loallyrespR}R_{ij}(x)=\frac{\ds e^{\alpha_j-\beta_j x_j}}{\ds\summ_{k\in\mc E_i}e^{\alpha_k-\beta_k x_k}+\chi_{\mc S}(i)}\,,\qquad j\in\mc E_i\,,\ee
where $\alpha,\beta\in\R^n$ are vectors such that $\beta\ge0$, and $\chi_{\mc S}$ is the indicator function of $\mc S$ in $\mc I$, i.e., $\chi_{\mc S}(i)=1$ if $i\in\mc S$ and $\chi_{\mc S}(i)=0$ if $i\notin\mc S$. As shown in \cite{ComoPartITAC13}, the dynamical flow network \eqref{eq:Rx} with locally responsive feedback routing $R(x)$ is globally monotone on $\R_+^n$. Points (iv) and (ii) of Theorem \ref{theo:stability-monotone} then imply a dichotomy: either the dynamical flow network \eqref{eq:Rx} admits a stable equilibrium $x^*$ (that is the limit of the solution started in $x(0)=0$) or it is unstable, i.e., all its solutions grow unbounded in time. Globally asymptotic stability of $x^*$ can then be established in the latter case (c.f.~\cite{ComoPartITAC13}). 
\end{example}
\begin{example}\label{ex3} In the work \cite{Como.Lovisari.ea:TCNS15}, routing policies such as those in Example \ref{ex2} are coupled with flow control policies whose role consists in reducing the flow towards downstream cells when  these are more congested than the upstream ones. Specifically, locally responsive feedback routing policies $R(x)$ satisfying properties (a), (b), and (c) of Example \ref{ex2} are considered together with flow controls $\Gamma(x)=\diag(\gamma(x))$ with the following properties: 
\begin{enumerate}
\item[(d)] each $\gamma_i(x)$ depends on local information only, consisting of the mass $x_k$ in the cells $k\in\mc E_i\cup\{i\}$; 
\item[(e)] the dependence on such local information is monotone in the sense that  
$$\frac{\partial}{\partial x_j}\gamma_i(x)\le0\,,\qquad\frac{\partial}{\partial x_k}\left(\gamma_i(x)R_{ij}(x)\right)\ge0\,,\qquad j,k\in\mc E_i\,,\quad j\ne k;$$
\item[(f)] if $x_j\to+\infty$ for all $j\in\mc E_i$ and $x_i$ remains bounded, then
$\lim \gamma_{i}(x)=0$; 
\item[(g)] if $x_i\to+\infty$ and $x_j$ remains bounded for some $j\in\mc E_i$, then
$\lim \gamma_{i}(x)=1$. 
\end{enumerate}
Specific forms include the $i$-logit routing policies \eqref{loallyrespR} coupled with the locally responsive flow controls 
\be\label{responsivegamma}\gamma_i(x)=\frac{\ds\summ_{k\in\mc E_i}e^{\alpha_k-\beta_k x_k}+\chi_{\mc S}(i)}{\ds e^{\alpha_i-\beta_i x_i}+\summ_{k\in\mc E_i}e^{\alpha_k-\beta_k x_k}+\chi_{\mc S}(i)}\,.\ee
As shown in \cite{Como.Lovisari.ea:TCNS15}, conditions (ii) and (v), along with the fact that the demand functions $\varphi(x)$ are nondecreasing imply that \eqref{eq:RGamma} is a monotone dynamical flow network. Using arguments analogous to those summarized in Theorem \ref{theo:stability-monotone}, a dichotomy similar to that mentioned in Example \ref{ex2} is then proved in \cite{Como.Lovisari.ea:TCNS15}:  either all solutions of \eqref{eq:RGamma} grow unbounded in time, or they all converge to a globally asymptotically stable equilibrium. 
\end{example}
Dynamical flow network models with demand constraints on the outflows such as the ones listed above do capture some congestion effects, however they do not account for spillbacks and upstream propagation of perturbations. These can be brought into the model through the introduction of supply functions $\sigma_i(x_i)$ that limit the maximum inflow $u_i+\sum_jf_{ji}$ possible in cells $i\in\mc I$ given their current mass $x_i$. Throughout the paper, the supply functions $\sigma_i(x_i)$ are assumed to be Lipschitz continuous, non-increasing (the more mass, the less  inflow can be accommodated in the cell), and concave (the more mass the more inflow decrement). We also denote by $$\ov x_i=\sup\{x_i\ge0:\,\sigma_i(x_i)>0\}$$ the (possibly infinite) \emph{buffer capacity} of cell $i$. (See the rightmost plot in Figure \ref{fig:demand-supply}.) Upon stacking up the supply functions in an $n$-dimensional vector $\sigma(x)=(\sigma_i(x_i))_{i\in\mc I}$ the supply constraints can be written in the compact form 
\be\label{eq:supply-constraint}u+R^Tz\le \sigma(x)\,.\ee 
Dynamical flow network models with both demand and supply constraints were  introduced by Daganzo \cite{Daganzo:94,Daganzo:95}, as discretizations of the Lighthill-Whitham-Richards PDE model of road traffic flow \cite{LighthillTrafficPTRS55,Richards:56}.  Observe that the supply constraint \eqref{eq:supply-constraint} together with the demand constraint \eqref{eq:demand-constraint} imply that the set $$\mc X=\{x\in\R^n:\,0\le x_i\le\ov x_i\,,\ \forall i\in\mc I\}$$ is invariant for the network flow dynamics \eqref{eq:dynRz}. 

\begin{example}\label{ex4} In the continuous-time version of Daganzo's cell transmission model studied, e.g., in \cite{Jabari.Liu:12,Grandinetti.ea:ECC15,Coogan.Arcak:TAC15,Coogan.Arcak:Automatica16,Como.ea:TRB16}, the network flow dynamics read
\be\label{eq:Gx}\dot x=u-(I-R^T)\Gamma(x)\varphi(x)\,,\ee
where the inflow vector $u$ and the routing matrix $R$ are exogenous and possibly time-varying, while $\Gamma(x)=\diag(\gamma_i(x))$ with the flow control parameters 
\be\label{FIFO}\gamma_i(x)=\sup\left\{\alpha\in[0,1]:\,\max_{k\in\mc E_i}\left(\alpha\sum_{h\in\mc I}R_{hk}\varphi_h(x_h)-\sigma_k(x_k)\right)\le0\right\}\,,\qquad i\in\mc I\,.\ee
Such flow control parameters ensure that both supply and demand constraints are met and that a first-in first-out (FIFO) rule is followed at junctions so that, in particular, the outflow from any cell $i$ is split among its immediately downstream cells $j\in\mc E_i$ exactly according to the split ratios $R_{ij}$ in every circumstances. Observe that in the \emph{free-flow region} \be\label{freeflow}\mc F:=\{x\in\R_+^n:\,u+R^T\varphi(x)\le\sigma(x)\}\,,\ee one has that $\gamma_i(x)=1$ for every cell $i$, so that the dynamics \eqref{eq:Gx} coincide with \eqref{eq:R} of Example \ref{ex1}. This is not the case outside the free-flow region, where the supply constraints \eqref{eq:supply-constraint} make $\gamma_i(x)<1$ for at least one cell $i$. In fact, observe that the cell-transmission model \eqref{eq:Gx} with FIFO diverge  rule \eqref{FIFO} is a monotone dynamical flow network in a domain $\mc D$ that includes the free-flow region $\mc F$, where its dynamics coincide with the ones of Example \ref{ex1}, while it is not necessarily monotone on the whole state space $\R_+^n$, except for special network topologies (c.f.~\cite{Coogan.Arcak:Automatica16}). 
Provided that sufficient conditions for the existence of an equilibrium $x^*$ in the free-flow region analogous to those in Example \ref{ex1} are satisfied, local asymptotic stability can be established using  point (v) of Theorem \ref{theo:stability-monotone} (c.f., \cite{Coogan.Arcak:TAC15}). For special network topologies for which \eqref{eq:Gx}-\eqref{FIFO} is globally monotone on $\R_+^n$ (essentially, networks with no diverge junctions) global asymptotic stability can be deduced from point (iv) of Theorem \ref{theo:stability-monotone}. We refer the interested reader to \cite{Coogan.Arcak:Automatica16}  for stability results of the cell-transmission model with FIFO diverge rule on a certain class of network topologies (polytree networks) where the resulting  dynamics are not generally globally monotone. 
\end{example}

\begin{example}\label{ex5} A variant of Daganzo's cell transmission model with non-FIFO diverge rule was considered, e.g., in \cite{Lovisari.Como.ea:CDC14,Karafyllis.Papageorgiou:TCNS15}. 
Here, the dynamics is more easily expressed in the form \be\label{xFx}\dot x=u+F^T(x)\1-F(x)\1-w(x)\,,\ee
 where, for a given exogenous routing matrix $\ov R$, the flows are given by 
\be\label{FXX}F_{ij}(x)=\gamma_{ij}(x)\ov R_{ij}\varphi_i(x_i)\,,\qquad\gamma_{ij}(x)=\sup\left\{\xi\in[0,1]:\,\xi\cdot\sum_{h\in\mc I}\ov R_{hj}\varphi_h(x_h)\le\sigma_j(x_j)\right\}\,,\ee
and $w_i(x)=(1-\sum_j\ov R_{ij})\varphi_i(x_i)$.
Observe that, in the free-flow region \eqref{freeflow}, the dynamics above coincide with the ones of Examples \ref{ex1} and \ref{ex4}, while this is not the case elsewhere. In fact, it can be proved that \eqref{xFx}-\eqref{FXX} is a monotone dynamical flow network on the whole state space: from this, under certain sufficient conditions, global stability results can be deduced, as illustrated in \cite{Lovisari.Como.ea:CDC14}.\end{example}



 Interestingly, it has been pointed out that, while not necessarily monotone outside the free-flow region, the cell transmission model with FIFO diverge rule  \eqref{eq:Gx}-\eqref{FIFO} of Example \ref{ex4} enjoys a \emph{mixed monotonicity} property, i.e., the entries of its Jacobian have signs that are constant throughout the state space. In fact, such mixed monotonicity property carries over to versions of the cell-transmission model that combine FIFO with non-FIFO diverge rules as in Example \ref{ex6}. (C.f.~\cite{Coogan.Arcak:Automatica16,Coogan.ea:2016}.)

%
%
%




\section{Robustness of dynamical flow networks and their margin of resilience} 
\label{sec:resilience}

In the series of papers \cite{ComoPartITAC13,ComoPartIITAC13,Como.Lovisari.ea:TCNS15,Savla.ea:TNSE14}, robustness of nonlinear dynamical flow networks with decentralized routing and flow controls is analyzed. The main idea consists in studying the impact that perturbations of the external inflows and of the cell flow capacities  have on the dynamic behavior of the flow network. In order to quantitatively measure such  robustness, the notion of \emph{margin of resilience} is introduced. This is defined as the smallest magnitude of a perturbation that drives the dynamical flow network to instability, such magnitude being measured as the sum of the aggregate loss of cell flow capacities and the aggregate increase of external inflows. In this section, we briefly review some versions of these results.

\begin{figure}
\begin{center}
\includegraphics[width=7.5cm]{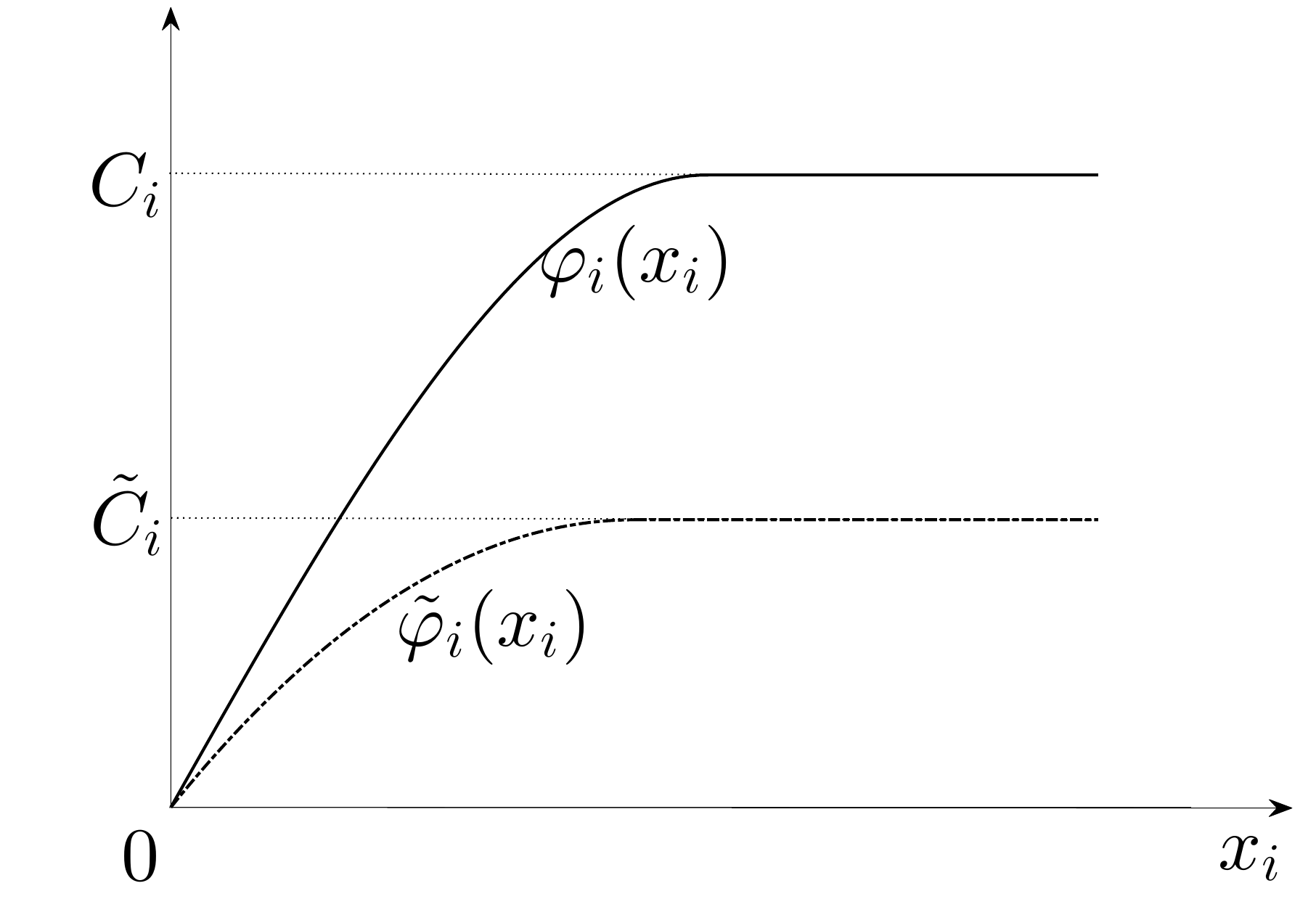}\end{center}
\caption{\label{fig:4} Considered perturbation affect the external inflow vector and the demand functions. Above a nominal demand function $\varphi_i(x_i)$ in a cell $i$ is plotted along with the perturbed demand function $\tilde\varphi_i(x_i)$ in the same cell. The magnitude of a perturbation is measured as in \eqref{delta}.} 
\end{figure}

In the interest of clarity, we will consider dynamical flow networks with constant external inflows and demand constraints only, while referring to the works \cite{ComoPartIITAC13,Como.Lovisari.ea:TCNS15,Savla.ea:TNSE14} for some results with finite buffer capacities that can be interpreted as special cases of supply constraints, as well as some discussion of extensions to time-varying external inflows. 
We start with a dynamical flow network
 \be\label{dynGR}\dot x=u-(I-R^T(x))\Gamma(x)\varphi(x)\,,\ee
 with topology $\mc T=(\mc I,\mc A,\mc R,\mc S)$, constant external inflows $u$,  demand functions vector $\varphi(x)$, routing matrix $R(x)$, and flow control  $\Gamma(x)=\diag(\gamma(x))$. 
We then consider both perturbations $\tilde u$ of the inflow vector $u$ and perturbations $\tilde\varphi(x)$ of the demand functions $\varphi(x)$ (see Figure \ref{fig:4}) and measure the \emph{magnitude} of such perturbations as 
\be\label{delta}\delta:=\sum_{i\in\mc R}|\tilde u_i-u_i|+\sum_{i\in\mc I}||\varphi_i(\,\cdot\,)-\tilde\varphi_i(\,\cdot\,)||_{\infty}\,.\ee
We then propose the following quantitative measure of robustness for dynamical flow networks. 
\begin{definition}\label{def:resilience}
The \emph{margin of resilience} $\nu$ of a dynamical flow network \eqref{dynGR}
with topology $\mc T=(\mc I,\mc A,\mc R,\mc S)$, constant external inflows $u$,  demand functions $\varphi(x)$, routing matrix $R(x)$, and flow control $\Gamma(x)$,  
is the infimum of the magnitudes $\delta$ of all perturbations $(\tilde u,\tilde\varphi(x))$ of $(u,\varphi(x))$ that make the perturbed dynamical flow network 
\be\label{pertnetwork}\dot x=\tilde u+(I-R^T(x))\Gamma(x)\tilde\varphi(x)\ee
unstable, i.e., such that the state $x(t)$ is unbounded in time for at least one initial condition $x(0)\in\R_+^n$. 
\end{definition}
A few comments are in order. First, Definition \ref{def:resilience} captures a \emph{worst case} notion of \emph{robustness}. In fact, one might think of the margin of resilience $\nu$ as the minimum effort required by a hypothetical adversary in order to modify the external inflows and demand functions so as to make the dynamical flow network unstable, provided that such effort is measured in terms of the perturbation magnitude \eqref{delta}. 

Second, notice from the righthand side of \eqref{pertnetwork} that the (feedback) routing $R(x)$ and flow control $\Gamma(x)$ are required to stabilize the system for any perturbation $(\tilde u,\tilde\varphi(x))$ affecting the external inflows and demand functions of magnitude smaller than the margin of resilience, without knowing the perturbation itself, but simply by measuring its effect on the state $x$. This is in line with classical robust control formulations. On the other hand, the considered notion of stability requires the state to remain bounded in time for every initial condition $x(0)$, not necessarily to converge to a pre-specified equilibrium. In fact, as we have seen in Section \ref{sec:demand-supply}, most of the considered routing and flow controls are such that the dynamical flow network is monotone and, when stable, it has a globally asymptotically stable equilibrium. Hence, in these cases, when the perturbed dynamical flow network remains stable, it admits a globally asymptotically stable equilibrium: the specific value of such equilibrium $\tilde x^*$ depends on the perturbation itself as well as on the routing and flow controls, but it does not affect the notion of margin of resilience.   

Observe that the margin of resilience depends on several factors: the network topology $\mc T$, the unperturbed demand functions $\varphi(x)$ (hence, in particular, the maximum flow capacities $C$), the external inflows $u$,  as well as the routing $R(x)$ and flow control $\Gamma(x)$. 
In fact, an upper bound on such margin of resilience can be obtained that depends only on the network topology, the maximum flow capacities, and the unperturbed external inflows: such upper bound holds true for every choice of routing and flow control policies. In order to state it, some further graph-theoretic notation is needed. For a subset of cells $\mc J\subseteq\mc I$, let 
$\mc K_{\mc J}\subseteq\mc I$ be the set of cells that either belong to $\mc J$ or are no longer outflow-connected in the network topology obtained from $\mc T$ by removing the cells in $\mc J$. In other terms, the set $\mc K_{\mc J}$ includes $\mc J$ as well as all those cells that are connected to the set $\mc S$ only though paths passing through $\mc J$. (See Figure \ref{fig:cut}.) Then, define the \emph{min-cut residual capacity} as 
\be\label{def:min-cut}C_{\text{min-cut}}=\min_{\mc J\subseteq\mc I}\left\{\left(\sum\nolimits_{j\in\mc J}C_j-\sum\nolimits_{k\in\mc K_{\mc J}}u_k\right)\vee0\right\}\,.\ee

\begin{figure}
\begin{center}
\includegraphics[width=8.5cm]{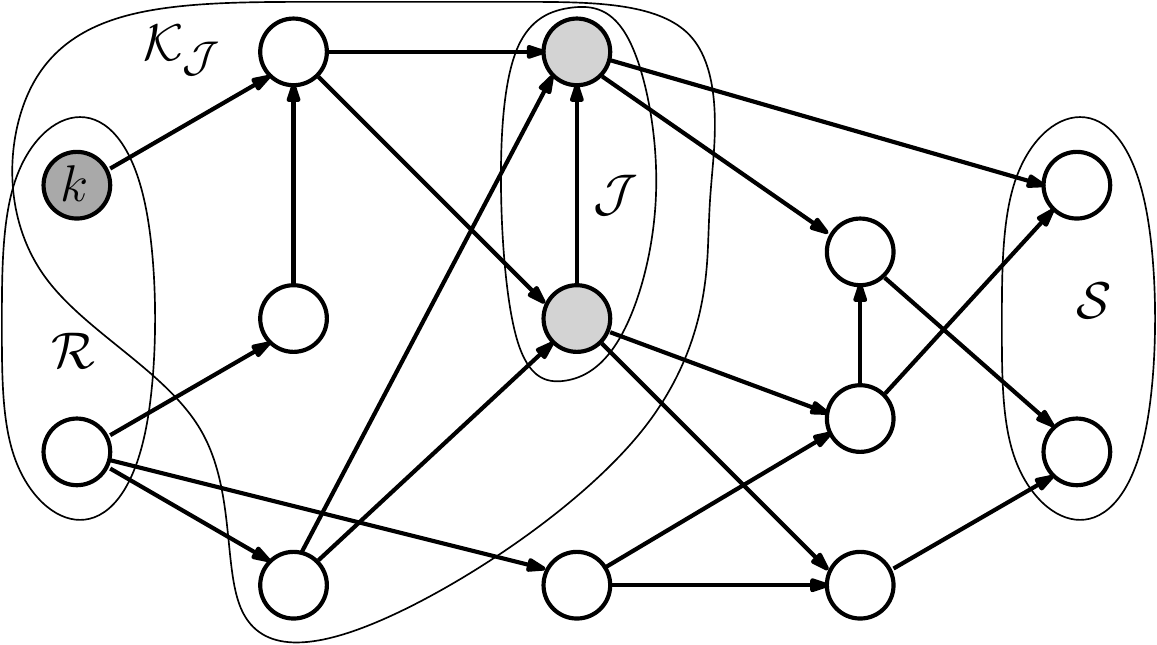}\end{center}
\caption{\label{fig:cut} The residual capacity of a cut $\mc J\subseteq\mc I$ is the aggregate capacity of the cells $j\in\mc J$ minus the total external inflow in the cells $k\in\mc K_{\mc J}$ where $\mc K_{\mc J}$ is the set of cells that either belong to $\mc J$ or are no longer outflow-connected if the cells in $\mc J$ are removed. In the network topology displayed above, $\mc J$ consists of two cells (colored in light grey), while $\mc K_{\mc J}$ consists of five cells. In this case, since $\mc K_{\mc J}\cap\mc R=\{k\}$, where $k$ is the cell colored in dark grey, the total external inflow in $\mc K_{\mc J}$ corresponds to the inflow $u_k$ in cell $k$.} 
\end{figure}

As stated in the following theorem, the min-cut residual capacity provides an upper bound 
\be\label{min-cutUB}\nu\le C_{\text{min-cut}}\ee
on the margin of resilience of any dynamical flow network with demand constraints \eqref{dynGR} with arbitrary routing and flow control policies. The inequality \eqref{min-cutUB} can be proved using a mass conservation argument that follows from the inherent structure of dynamical flow networks \eqref{compact} and can be thought of as the analogue of a classical result for static network flows. In fact, for such static network flows,  tightness of the min-cut bound is guaranteed by the celebrated max-flow min-cut theorem \cite{FordFulkersonCJM56,EliasIRETIT56}, which can be interpreted as a linear-programming duality result. For dynamical flow networks, the dynamic constraints pose additional challenges to the achievement of the maximum margin of resilience, i.e., tightness of the upper bound \eqref{min-cutUB}. The following theorem states that such maximum margin of resilience is achieved on arbitrary network topologies by the combined locally responsive feedback routing and flow control policies introduced in Example \ref{ex3}. In contrast, it is proved that not only the fixed routing policies of Example \ref{ex1} but also the locally responsive feedback routing without flow control of Example \ref{ex2} fall short of achieving maximum margin of resilience on arbitrary network topologies.  
%

\begin{theorem}\label{theo:resilience} Let $\mc T=(\mc I,\mc A,\mc R,\mc S)$ be an outflow-connected network topology, $u$ a constant external inflow vector supported on $\mc S$, and $\varphi(x)$ a vector of strictly increasing demand functions. Consider a dynamical flow network \eqref{dynGR} with routing policy $R(x)$ and flow control $\Gamma(x)$, and let $\nu$ be its margin of resilience. Then, 
\begin{enumerate}
\item[(i)] the upper bound \eqref{min-cutUB} holds true. 
\end{enumerate}
Moreover: 
\begin{enumerate}
\item[(ii)] for locally responsive routing $R(x)$ and flow control $\Gamma(x)$ as in Example \ref{ex3}, 
$$\nu=C_{\text{min-cut}}\,,$$
i.e., the margin of resilience coincides with the min-cut residual capacity.
\end{enumerate}
On the other hand, if the network topology $\mc T$ is such that $\mc G=(\mc I,\mc A)$ is the line-digraph of an acyclic digraph $\mc N=(\mc V,\mc I)$ (as per Remark \ref{ft}), then
\begin{enumerate}
\item[(iii)] if the routing matrix $R$ is constant and there is no flow control, i.e., $\Gamma=I$, then the margin of resilience coincides with the minimum cell residual capacity 
\be\label{resilience1} \nu=\min_{i\in\mc I}\{(C_i-z_i^*),\vee0\}\,,\qquad z^*=(I-R^T)^{-1}u\,;\ee
\item[(iv)] if $R(x)$ is a locally responsive feedback routing policy with no flow control as in Example \ref{ex2}, the margin of resilience satisfies 
\be\label{resilience2} \nu=\min_{i\in\mc I}\left\{\sum_{j\in\mc E_i}(C_j-z_j^*)\right\}\,,\ee
where $z^*=\varphi(x^*)$ with $x^*$ being the equilibrium if the unperturbed network is globally asymptotically stable, and $z^*=C$ otherwise.  
\end{enumerate}
\end{theorem}
\begin{proof}
Part (i) can be proved as in \cite{ComoPartITAC13}. Part (ii) follows from \cite[Theorem 1(i)]{Como.Lovisari.ea:TCNS15}. Part (iv) is proved in \cite{ComoPartIITAC13}, while part (iii) can be proved with a similar (and in fact simpler) argument. 
\qed\end{proof}\medskip

While Part (i) of Theorem \ref{theo:resilience} captures the physically intuitive fact that a dynamical flow network cannot be stable if the external inflow violates some cut capacity constraint, Part (ii) states quite a remarkable result in that a global objective (maximum margin of resilience) is achievable by completely decentralized routing and flow control policies that use no global knowledge whatsoever of the network state or structure. That  this is nontrivial is confirmed by Parts (iii) and (iv) of Theorem \ref{theo:resilience} stating that, if there is no flow control and the routing is either fixed or locally responsive, then the margin of resilience coincides with, respectively, the minimal cell residual capacity \eqref{resilience1}, and the minimum out-neighborhood resilience capacity \eqref{resilience2}, quantities that have been proved to be arbitrarily far away from the min-cut residual capacity \eqref{def:min-cut}: see \cite{ComoPartITAC13,ComoPartIITAC13}. As explained in \cite{Como.Lovisari.ea:TCNS15}, the reason why the combination of locally responsive feedback flow control and routing is able to achieve maximal resilience resides in the ability to properly propagate congestion ---and hence, implicitly, information--- in a controlled way, both upstream and downstream, thanks to the flow control, and to reroute flow towards less congested paths (thanks to the routing). In fact, Part (ii) of Theorem \ref{theo:resilience} should be compared with results available for packet-switched data networks, e.g., maximal throughput and stability of the back-pressure algorithm \cite{TassiulasTAC93}. In contrast, locally responsive routing by itself, as in Example \ref{ex2} is not sufficient to achieve maximal margin of resilience since it is not able to propagate information upstream but just to reroute flow in a locally optimal way, whereas fixed routing as in Example \ref{ex1} does even worse because it fails to reroute flow towards less congested paths.

\section{Conclusion} 
In this paper, we have surveyed stability and robustness properties of dynamical flow networks. These are a family of dynamical systems derived from mass conservation laws on directed graphs with exogenous inflows and outflows towards the external environment. After reviewing some structural stability properties of linear dynamical flow networks, we have focused on an important class of nonlinear dynamical flow networks characterized by an additional monotonicity property. While rich enough to encompass many examples of dynamical flow networks studied in the literature, including some well-established models for road traffic networks, the class of monotone dynamical flow networks has enough structure to allow for tractable analysis and scalable design. In particular, we have presented a non-expansiveness result in the $l_1$-metric for solutions of mononotone dynamical flow networks and applied it to their (differential) stability analysis. We have then focused on dynamical flow networks with cell demand constraints and introduced the margin of resilience as a quantitative measure of their robustness. We have shown that maximal margin of resilience on a given capacitated network flow topology coincides with a well-known graph-theoretical quantity, the min-cut residual capacity, and that it can be achieved by a combination of decentralized locally responsive feedback flow control and routing policies that use local information only and require no global knowledge of the network. 

It is worth emphasizing that the treatment in this paper has been focused on single-commodity first-order models of dynamical flow networks and specifically on their stability and robustness properties. Other aspects of dynamical flow networks of current interest that have not been discussed here include: optimal dynamical flow network control  (see, e.g., \cite{Como.ea:TRB16} for convex formulations with demand and supply constraints generalizing ideas in \cite{Ziliaskopoulos:00,GomesTRC06}); multi-scale models coupling dynamical flow networks with game-theoretic learning dynamics \cite{Como.Savla.ea:SICON13} modeling the selfish route choice behaviors of the drivers; extensions of the margin of resilience framework to multi-commodity dynamical flow networks for which only preliminary results are currently available \cite{Nilsson.ea:2014}. 

\section*{Acknowledgements}
The author is grateful to Professor Ketan Savla of the University of Southern California for a long-standing collaboration on the topics surveyed in this paper. He would also like to thank his other coauthors Enrico Lovisari, Munther A.~Dahleh, Daron Acemoglu, Emilio Frazzoli, and Anders Rantzer. Finally, the author wishes to thank Professor Tryphon Georgiou for inviting him to present these results at the 22nd International Symposium on Mathematical Theory of Networks and Systems, and Professor Francoise Lamnabhi-Lagarrigue for soliciting a contribution in this venue. 


\section*{References}

\bibliographystyle{elsarticle-num} 
\bibliography{bibliography2,bibliography}

\end{document}